\newif\ifarxiv\arxivtrue
\newif\ifsagt\sagtfalse
\documentclass[runningheads]{llncs}
\usepackage{graphicx}
\usepackage{amsthm}
\usepackage{amsmath}
\usepackage{amssymb}
\usepackage{cleveref} 
\usepackage{mathtools} 
\usepackage{braket} 
\usepackage{thmtools} 
\usepackage{thm-restate} 

\newcommand{\R}{\mathbb{R}}

\newcommand{\N}{\mathbb{N}}

\newcommand{\bigO}{\mathcal{O}}
\newcommand{\abs}[1]{\left\lvert #1 \right\rvert}

\newcommand*\diff{\mathop{}\!\mathrm{d}} 
\renewcommand{\epsilon}{\varepsilon} 

\renewcommand{\l}{\ell} 

\newcommand{\Paths}{\mathcal{P}}

\usepackage[colorinlistoftodos]{todonotes}

\renewcommand{\paragraph}[1]{\vspace{0.5em}\noindent\textbf{#1}\;}

\newcommand{\trg}{(G,s,t,\allowbreak r_0,\tau_e,\allowbreak\nup{e},\num{e},\allowbreak\sigma_e,M)} 

\newcommand{\mtpoa}[1]{\text{PoA}(#1)} 
\newcommand{\comp}[1]{T_{#1}} 
\newcommand{\opt}{T_\text{opt}} 
\newcommand{\eq}{T_\text{EQ}} 
\newcommand{\br}[1]{\text{BR}(#1)} 
\newcommand{\allf}[1]{\mathcal{F}({#1})} 
\newcommand{\eact}[1]{E^\prime_{#1}} 
\newcommand{\eres}[1]{E^\star_{#1}} 
\newcommand{\efull}[1]{\bar{E}_{#1}} 
\newcommand{\cspn}[1]{G^\prime_{#1}} 
\newcommand{\ome}[1]{\l_t(\theta_{#1})} 

\newcommand{\fp}[1]{f^+_{#1}} 
\newcommand{\fm}[1]{f^-_{#1}}

\newcommand{\fh}{y} 
\newcommand{\gp}[1]{g^+_{#1}}

\newcommand{\xp}[1]{x^\prime_{#1}}
\newcommand{\lp}[1]{\l^\prime_{#1}}
\newcommand{\lpp}[2]{\l_{#1,#2}^\prime} 
\newcommand{\nup}[1]{\nu^+_{#1}}
\newcommand{\num}[1]{\nu^-_{#1}}
\newcommand{\bm}[1]{b^-_{#1}}
\newcommand{\bp}[1]{b^+_{#1}}
\newcommand{\qp}[1]{q^\prime_{#1}} 
\newcommand{\qpp}[2]{q_{#1,#2}^\prime} 
\newcommand{\phases}{\mathcal{I}_{f}} 
\newcommand{\ca}[1]{\kappa_{#1}} 
\newcommand{\minc}[2]{\mathsf{c}^{#1}_{#2}} 

\newcommand{\inset}[1]{\delta^-(#1)} 
\newcommand{\outset}[1]{\delta^+(#1)} 

\begin{document}
\title{The Impact of Spillback on the Price of Anarchy for Flows Over Time\thanks{Funded by the Deutsche Forschungsgemeinschaft (DFG, German Research Foundation) under Grant BR 4744/2-1 and Germany's Excellence Strategy – The Berlin Mathematics Research Center MATH+ (EXC-2046/1, project ID: 390685689).}}

%
%
\author{Jonas Israel\orcidID{0000-0002-3992-3203} \and
Leon Sering\orcidID{0000-0003-2953-1115}}
\authorrunning{J. Israel and L. Sering}
%
\institute{Technische Universität Berlin, Germany \\ \email{j.israel@tu-berlin.de} \qquad \email{sering@math.tu-berlin.de} 
}

\maketitle              
\begin{abstract}
Flows over time enable a mathematical modeling of traffic that changes as time progresses. In order to evaluate these dynamic flows from a game theoretical perspective we consider the \emph{price of anarchy} (PoA).
In this paper we study the impact of spillback effects on the PoA, which turn out to be substantial. It is known that, in general, the PoA is unbounded in the spillback setting. We extend this by showing 
that it is still unbounded even when considering networks with unit edge capacities and 
that 
the \emph{Braess ratio} can be arbitrarily large.

In contrast to that, we show that on a fixed network the PoA as a function of the flow amount is bounded by a constant and also upper bound the PoA for the set of networks where the outflow capacities satisfy certain constraints depending on the quickest flow. This upper bound only depends on the worst spillback factor of the Nash flows over time of the given network. It therefore provides a way to quantify the impact of spillback to the quality of the dynamic equilibria.

In addition, we show the surprising fact that the introduction of spillback behavior can actually speed up dynamic equilibria in some networks.

\ifsagt 
\\[1em]
\textbf{Related Version:} A full version of this paper including all proofs is available at
\url{https://arxiv.org/abs/????.?????}.
\fi
\end{abstract}

\section{Introduction} Road traffic is an integral part of modern societies,
which consists of many users with individual behaviors and goals. For this reason traffic dynamics are very hard
to predict and can barely be controlled. However, through recent technologies such as intelligent navigation systems it
might be possible to positively affect the behavior of traffic, steering it towards shorter travel
times leading to less pollution and an overall improved quality of life.

In the following research work we focus on a mathematical traffic flow model called \emph{flows over time with
spillback}. Here, the network is depicted as a graph with a source and a sink, and the traffic flow can progress in
continuous time from one vertex over an edge to the next vertex. We consider flow as
a continuous stream affected by two types of temporal factors. First, flow does not travel instantaneously through the
network but needs actual time to traverse an edge, and second, flow on an edge may change over time. Compared to
static network flows these temporal components enable us to model traffic realistically through different congestion levels.
To model road constraints within the network, we equip each edge with an inflow and an outflow capacity governing with
which rate flow can enter and leave the edge, a length characterizing the time it takes a flow-particle to travel from
the tail to the head of the edge, and finally, a storage capacity which describes how much flow volume fits on the
edge. If the desired outflow exceeds the outflow capacity of an edge the excess flow queues up in front of the
bottle-neck at the head of the edge. If at any point in time the queue of an edge is so large that the amount of flow traversing the edge plus the amount of flow in the queue equals the storage capacity, the edge is considered full and new flow can only enter if
at least as much flow leaves at the same time. With this mechanic it is possible to model \emph{spillback}, i.e.,
the phenomenon that traffic congestion at one street can block exits or intersections further upstream. The ability to model spillback within the framework of flows over time is a very
recent discovery~\cite{sering19}, which has not been studied much yet.

As we experience in our everyday lives traffic is not performing optimal most of the time, but rather consists of agents
that behave egoistically.
Thus, we are interested in game theoretic aspects of this flow model, particularly in the \emph{price of anarchy} (PoA), the ratio of the worst uncoordinated behavior described via a dynamic equilibrium, and the optimal flow behavior measured by some social cost function. In real-world scenarios that ratio could give us an idea of how much one can possibly improve traffic through optimized traffic control, for example through modern navigation systems or autonomous driving.
Even though it has been shown in~\cite{sering19} that the PoA in networks with spillback is unbounded in general we investigate the dependency of the PoA on several
parameters, for example, the minimal spillback factor, which measures how much the
capacities of an edge are reduced due to spillback.
Another interesting phenomenon of selfish road users we study is the well known Braess paradox~\cite{braess68}.
It states that the overall travel time of all users might decrease if a frequently used road segment gets closed. In
reverse, this means that building new roads between heavily used section of the network might cause more congestion and
longer travel times.

\paragraph{Related work.}
Flows over time were first introduced by Ford and Fulkerson~\cite{ford15} in the context of an
optimization problem to route as much flow as possible in a given time horizon. Gale~\cite{gale59} proved the existence of earliest arrival flows which optimize the amount of flow routed to the
sink simultaneously for all points in time and Wilkinson~\cite{wilkinson71} later presented an algorithm to compute these
flows. 
For an overview on flows over time from an optimization point of view we refer to the survey by Skutella~\cite{skutella09}.
From a game theoretic point of view, flows over time were first considered by Vickrey~\cite{vickrey69} in the setting of transportation research. In the last years the theory of Nash equilibria for flow
models has been advanced significantly. From the introduction of the price of anarchy by Koutsoupias and Papadimitriou
\cite{koutsoupias99,papadimitriou01} and the congestion games studied by Roughgarden and Tardos~\cite{roughgarden02,roughgarden05} (both for static flows), over existence results concerning the dynamic (i.e., time dependent) model by
Meunier and Wagner~\cite{meunier10}, to the constructive approach to dynamic equilibria by Koch and Skutella~\cite{koch11}. Here, the authors present a novel notion of dynamic equilibria, called \emph{Nash flows over time}, which enabled a whole set of proceeding research. This new research includes the study of existence, uniqueness and the long-term behavior of Nash flows over
time by Cominetti et al.~\cite{cominetti11,cominetti15,cominetti17}, the work by Macko et al.~\cite{macko10} about the Braess paradox for
flows over time as well as the extension to multi-terminal settings~\cite{sering18}. Of special interest to the paper at hand are the results by Bhaskar et al.~\cite{bhaskar15} and very recently by
Correa et al.~\cite{correa19} about the PoA for flows over time. Since it was already shown that the evacuation-PoA (maximizing the flow amount within some time horizon) is unbounded~\cite{koch11}, they focus on the time-PoA (minimizing the completion time for a given flow amount) for which they establish an upper bound of $\frac{e}{e-1}$ under some constraints on the capacities of the network.
Sering and Vargas
Koch~\cite{sering19} generalized the flows over time model in order to represent spillback and transferred the results about dynamic equilibria to this extension.
Very recently, Graf et al.~\cite{graf2020ide} characterized an alternative equilibrium concept for flows over time, where particles do not predict the future evolution of the flow but instead reconsider their route choice on every node.
In addition, there is a active research line on packet routing models, where traffic is represented by atomic vehicles that traverses the network in discrete time steps. Recent progress in this area is due to Cao et al.~\cite{cao2017network}, Scarsini et al.~\cite{scarsini2018dynamic}, Harks et al.~\cite{harks2018competitive} and Peis et al.~\cite{peis2018oligopolistic}.

\paragraph{Contribution and outline.}
We study the price of anarchy of flows over time with spillback introduced in~\cite{sering19}, which is known to be unbounded in general. After
introducing the model in \Cref{sec:model}, we show in  \Cref{sec:lower} that the PoA stays unbounded even if we
restrict the set of networks to a specific topology but allow arbitrary capacity, or in reverse if we only allow unit
capacities but therefore more complex graph structures.
Furthermore, we show that the
Braess ratio can be arbitrarily large depending only on the minimum edge capacity.
Even
though it seems that the addition of full edges and spillback only increases completion times this is not a general rule, as we show that there
are examples where the completion time of Nash flows over time is larger when disabling spillback.
In contrast to the above lower bounds we show in \Cref{sec:upper} that if we consider the case of temporal routing games on a fixed network, i.e., only the flow amount that gets routed through the network varies, the PoA is bounded by a constant. 
In the end we translate the ideas of~\cite{bhaskar15} to the model with spillback and prove an upper bound of $\frac{\minc{}{} e}{\minc{}{} e - 1}$ on the PoA in networks with specific conditions on the capacities in dependency of a maximal flow over time. This upper bound only depends on the worst spillback factor $\minc{}{}$ of the Nash flows over time of the given network, and therefore provides a way to quantify the impact of spillback to the PoA (note that $e$ denotes the Euler constant here).
Finally, we give a brief conclusion and outlook for further research in \Cref{sec:conclusion}.



\section{The Model} \label{sec:model}
In the following we want to recall the essential definitions of the flow over time model with deterministic queuing. We consider the extended version that handles 
spillback effects, as introduced in~\cite{sering19} and mainly stick to the same notation.

\paragraph{Flow dynamics.}
We consider a \emph{network} $\Gamma = (G, s, t, r_0, \tau, \nu^+, \nu^-, \sigma)$ given by a directed graph $G = (V, E)$ with a single \emph{source} $s$ and a
single \emph{sink} $t$, such that every vertex is reachable from $s$. We have a \emph{network inflow rate} of
$r_0 >0$ determining the constant rate of flow entering the network from time $0$ onward. Furthermore, every edge $e \in E$ is equipped with a \emph{transit
time} $\tau_e \geq 0$, an \emph{in- and outflow capacity} $\nu_e^+ > 0$ and $\nu_e^- > 0$ as well as a \emph{storage capacity} $\sigma_e > 0$. In order to
avoid undefined flow behavior, we require that traversing flow alone can never fill up an edge, i.e., $\sigma_e > \nu_e^+
\cdot \tau_e$ and that the total transit time of every directed cycle is strictly positive. For technical reason we furthermore assume that all properties are rational numbers.

A \emph{flow over time} is given by a family of locally integrable and bounded functions $f = (f_e^+, f_e^-)_{e \in E}$, where $f_e^+, f_e^- \colon \R_{\geq 0} \to \R_{\geq 0}$ denote the in- and outflow rate of edge $e$ at every point in time. The \emph{cumulative in-} and \emph{outflow} and the \emph{queue size} are given by
\[F_e^+(\theta) \coloneqq \!\int_0^\theta f_e^+(\xi) \diff \xi, \quad F_e^-(\theta) \coloneqq\! \int_0^\theta f_e^-(\xi) \;\; \text{ and } \;\; z_e(\theta) \coloneqq F_e^+(\theta - \tau_e) - F_e^-(\theta).\]

We require that flow is preserved at every edge $e$ (\emph{non-deficit constraint}) and at every vertex $v \in V \setminus \set{t}$ (\emph{conservation constraint}), which means, for every point in time $\theta$ we have \vspace{-0.8em}
\[z_e(\theta) \geq 0 \qquad \text{ and } \qquad \sum_{e \in \delta_v^+} f_e^+(\theta) - \sum_{e \in \delta_v^-} f_e^-(\theta) = \begin{cases} 0 & \text{ for } v \in V \setminus \set{s, t},\\
r_0 & \text{ for } v = s.
\end{cases} \]
Here, $\delta_v^-$ is the set of all \emph{incoming} and $\delta_v^+$ the set of all \emph{outgoing} edges of node~$v$.
An edge $e$ is \emph{full} at time $\theta$ if the total amount of flow on $e$, called \emph{edge load}, $d_e(\theta) \coloneqq F_e^+(\theta) - F_e^-(\theta)$ reaches the storage capacity $\sigma_e$. The \emph{inflow bound} $b_e^+(\theta)$ denotes that current inflow capacity, which might be smaller than $\nu_e^+$ due to spillback, and the \emph{push rate} $b_e^-(\theta)$ specify the current \emph{desired} outflow rate, which is reached whenever there are no restrictions of following links. Formally, we have
\[b_e^+(\theta) \!\coloneqq \!\begin{cases}
\nu_e^+ &\hspace{-1.25cm} \text{ if } d_e(\theta) \!<\! \sigma_e\\
\min \set{f_e^-(\theta), \nu_e^+} & \text{else,}
\end{cases} 
\text{ and } \;\;b_e^-(\theta) \!\coloneqq\! \begin{cases}
\nu_e^- & \hspace{-1.2cm} \text{ if } z_e(\theta) \!>\! 0,\\
\min \set{f_e^+(\theta - \tau_e), \nu_e^-} & \text{else.}
\end{cases}\]

A flow over time $f$ is \emph{feasible} if for all edges $e$ and all times $\theta$ it satisfies $f_e^+(\theta) \leq
b_e^+(\theta)$ (\emph{inflow condition}) and if there exists a $c_v \in (0, 1]$ for every $v\in V$ such that for every $e\in \inset{v}$ $f_e^-(\theta) =
\min\set{b_e^-(\theta), \nu_e^- \cdot c_v}$ (\emph{fair allocation condition}). 
Furthermore, we require for all time $\theta$ that every
vertex $v$ with an incoming edge $e_1 \in \delta_v^-$ with $f_{e_1}^-(\theta) <
b_{e_1}^-(\theta)$ (called \emph{throttled edge}) there exists an outgoing edge $e_2 \in \delta_v^+$ with $f_{e_2}^+(\theta) =
b_{e_2}^+(\theta)$ (\emph{no-slack condition}). 
Finally, the set of full edges should be
cycle free at every point in time (\emph{no-deadlock condition}).

For a given $v\in V$ the maximal value $c$ that satisfies the fair allocation condition at a given point in time $\theta$ is called \emph{spillback factor} denoted by $c_v(\theta)$. This value denotes the reduction of the outflow capacity due to spillback leading to the \emph{effective outflow capacity} of $\nu_e^- \cdot c_v(\theta)$. If the outflow rate $f_{uv}^-(\theta)$ of an incoming edge is strictly smaller than the push rate $b_{uv}^-(\theta)$, this edge is throttled implying $c_v(\theta) < 1$, which means that there is spillback at $v$. In this case the no-slack condition ensures that there is a reason for the spillback in form of an outgoing exhausted edge $vw$: $f_{vw}^+(\theta) =
b_{vw}^+(\theta)$.
The spillback factor will play an important role throughout this paper. For more details and further intuition on the definitions of a feasible flow over time in this setting we refer to~\cite{sering19}.

\paragraph{Nash flows over time.}
In order to define Nash flows over time we need to define the arrival time of every particle of the flow. To simplify the notation we identify every particle with the point in time $\theta$ when it enters the network at the source.
For every edge $e$ we define the \emph{waiting time function} $q_e \colon \R_{\geq 0} \to \R_{\geq 0}$ by $q_e(\theta) \coloneqq \min \Set{q \geq 0 | \int_{\theta + \tau_e}^{\theta + \tau_e + q} f_e^-(\xi) \diff \xi= z_e(\theta + \tau_e) }$, i.e., $q_e(\theta)$ denotes the time a particle entering $e$ at time $\theta$ waits in the queue. For every vertex $v$ the \emph{earliest arrival time function} $\l_v \colon \R_{\geq 0} \to \R_{\geq 0}$ denotes the earliest point in time the particle $\theta$ (which enters the network at time $\theta$) can reach~$v$: 
\begin{align*}
\l_v(\theta) \coloneqq \begin{cases}
\quad \theta & \text{if } v = s,\\
\min\limits_{e = uv \in E} \l_u(\theta) + \tau_e + q_e(\l_u(\theta)) & \text{else.}
\end{cases}\end{align*}
For a given particle $\theta$ the \emph{current shortest path network} $G'_\theta = (V, E'_\theta)$ is the network of all edges $e = uv$ that are \emph{active for $\theta$}, i.e., for which $\l_v(\theta) = \l_u(\theta) + \tau_e + q_e(\l_u(\theta)).$
It contains all $s$-$v$-paths that particle $\theta$ can use to be at $v$ at the earliest possible point in time.
Furthermore, we denote the \emph{resetting edges} $E^*_\theta$ as the set of edges for which particle~$\theta$ encounters a queue when taking a current shortest path and $\bar E_\theta$ denotes the set of edges which are full when particle $\theta$ reaches its tail. More precisely, $E^*_\theta \coloneqq \set{ e = uv \in E | q_e(\l_u(\theta)) > 0}$ and $\bar E_\theta \coloneqq \set{e = uv \in E | d_e(\l_u(\theta)) = \sigma_e}$.

\begin{definition}[Nash flow over time]
We call a feasible flow over time $f$ a \emph{Nash flow over time}, or \emph{dynamic equilibrium}, if almost every particle uses a current shortest $s$-$t$-path, i.e., if $f_e^+(\theta) > 0$ implies $\theta \in \l_u(\Theta_e)$ for all $e = uv \in E$ and almost all $\theta \in \R_{\geq 0}$,
where $\Theta_e \coloneqq \set{\theta \in \R_{\geq 0} | e \in E'_\theta}$ denotes 
all particles for which $e$ is active.
\end{definition}

Equivalently, it has been shown~\cite[Lemma 4.1]{sering19} that a feasible flow over time is a dynamic equilibrium if and only if
$F^+_e(\l_u(\theta)) = F^-_e(\l_v(\theta))$ for all $e = uv \in E$ and  all $\theta \in \R_{\geq 0}$.
By setting $x_e(\theta) \coloneqq F^+_e(\l_u(\theta)) = F^-_e(\l_v(\theta))$ we observe that $(x_e(\theta))_{e \in E}$ form a static $s$-$t$-flow of value $r_0 \cdot \theta$. Since the $x_e$ are absolute continuous, their derivatives
\begin{equation}\label{eq:thin_flows}
x'_e(\theta) = f_e^+(\l_u(\theta)) \cdot \l'_u(\theta) = f_e^-(\l_v(\theta)) \cdot \l'_v(\theta)
\end{equation}
exist almost everywhere and can be seen as the strategy of particle $\theta$ (as for every $\theta$ it is a static $s$-$t$-flow of value $r_0$). For a fixed $\theta$ the derivatives $x'_e \coloneqq x'_e(\theta)$ and $\l'_v \coloneqq \l'_v(\theta)$ together with the spillback factors $c_v \coloneqq c_v(\l_v(\theta))$ are called \emph{spillback thin flows} and with $b_e^+ \coloneqq b_e^+(\l_u(\theta))$ for all $e = uv$ satisfy the following equations:
\begin{alignat*}{4}
\l'_s &= 1, &&\\
\l'_v &= \min_{e = uv \in E'_\theta} \rho_e\left(\l'_u, x'_e, c_v\right) \quad&& \text{ for } v \in V \setminus \set{s},\\
\l'_v &= \rho_e\left(\l'_u, x'_e, c_v\right)  && \text{ for } e = uv \in E'_\theta \text{ with } x'_e > 0,\\
\l'_v &\geq \max_{e = vw \in E'_\theta} \frac{x'_e}{b_e^+} && \text{ for } v \in V,\\
\l'_v &= \max_{e = vw \in E'_\theta} \frac{x'_e}{b_e^+} && \text{ for } v \in V \text{ with } c_v < 1,
\end{alignat*}
where \vspace{-0.5cm}
\[\rho_e(\l'_u, x'_e, c_v) \coloneqq \begin{cases}
\frac{x'_e}{c_v \cdot \nu_e^-} & \text{ if } e = uv \in E_\theta^*,\\
\max \Set{\l'_u, \frac{x'_e}{c_v \cdot \nu_e^-}} & \text{ if } e = uv \in E'_\theta \setminus E_\theta^*.
\end{cases}\]

It turns out that the particles of a Nash flow over time $f$ can be divided into intervals, so called \emph{phases},
for which the derivatives (and thus the inflow and outflow rates) stay constant. We denote the set of phases by $\phases$. The transition points between two
phases correspond to one or multiple \emph{events}: A new edge (and therefore new
$s$-$t$-paths) can become active, a queue can deplete, an edge can become full or the outflow rate (and hence the inflow
bound) of a full edge might change. Note however, that an event at edge $e = uv$ for a particle $\theta$ does not happen
at time $\theta$ itself but rather at time $\l_u(\theta)$ when the particle entering the network
at time $\theta$ reaches vertex $u$ (while taking a shortest $s$-$u$ path).

\paragraph{Games, optimal flows and the price of anarchy.}
For a \emph{temporal routing game} we consider a finite volume of flow $M \in (0, \infty)$ entering the network. For a Nash flow over time $f$ the last particle enters the network at time $\frac{M}{r_0}$ and leaves the network at time $\l_t(\frac{M}{r_0})$. As the network satisfy the first-in-first-out-principle (FIFO), $\l_t$ is non-decreasing, which means that $\comp{f} \coloneqq \l_t(\frac{M}{r_0})$ denotes the \emph{completion time} when the entire flow of volume $M$ has reached $t$. Most of the time we identify a network $\Gamma$ with its corresponding temporal routing game (i.e., $\Gamma$ and $M$).
%
In contrast to dynamic equilibria, optimal \emph{quickest flows} can be computed by determining a time horizon $T$ with $\l_t(\frac{M}{r_0}) = T$ by applying a binary search framework to the maximum flow over time problem. Hereby, a maximum flow over time for time horizon $T$ can be constructed via a feasible static flow $\fh$ maximizing $T \cdot \abs{\fh} - \sum_{e \in E} \tau_e \cdot \fh_e$. This \emph{underlying static
flow} $\fh$ is then temporally repeated, which means a rate of $\fh_p$ is sent into every $s$-$t$ path $p \in \Paths_{st}$ over time $[0, T - \tau_p]$. The arrival time of the last particle at $t$ (i.e., the optimal completion time) is denoted by $\opt(M)$. For more details on optimal flows over time we refer to Skutella's survey~\cite{skutella09}.

In this paper we consider the \emph{time price of anarchy} (which we simply refer to as ``price of anarchy''). For a given temporal routing game $\Gamma$ it measures the worst ratio between the arrival time at $t$ for the last particle in a Nash flow over time and the arrival time in an optimal flow: $\mtpoa{\Gamma} \coloneqq \frac{\comp{EQ}(\Gamma)}{\opt(\Gamma)}$.\footnote{All results from Section \ref{sec:upper} can also be translated to the \emph{total delay price of anarchy} measuring the arrival times of all particles combined, similarly as is done in~\cite{bhaskar15}.
}
As it is unknown whether the arrival time functions $\l_t$ are unique over all Nash flows over time, we need to consider the worst dynamic equilibrium, i.e., $\comp{EQ}(\Gamma) \coloneqq \sup_{f\in\allf{\Gamma}} \comp{f}$, where $\allf{\Gamma}$ denotes the set of Nash flows over time in $\Gamma$.

\paragraph{Further notation.}
We enumerate the event points by the order of their occurrence seen by particles at the source, i.e., $\theta_i < \theta_{i + 1}$ and say phase $i$ is given by $(\theta_{i -1}, \theta_i)$ (using $\theta_0 = 0$).\footnote{We imagine $i$ as a natural number. But since it is an open
question whether the event point converges to a finite limit, it is possible to expand the index set to the ordinal numbers up to $\omega^\omega$. In this case the $i$-th phase should be defined as $(\theta_i, \theta_{i + 1})$ as it is not possible to determine a predecessor of an ordinal number. For the sake of simplicity however, we stick to the definition where $(0, \theta_1)$ is the first phase.}
In addition, we consider the point in time $\frac{M}{r_0}$ when the last particle enters the network as the last event $r$, i.e., $\theta_r \coloneqq \frac{M}{r_0}$.
Since the edge sets $E'_\theta$, $E^*_\theta$, $\bar E_\theta$, the inflow bound, and hence, the spillback thin flow stay constant within each phase $i$ we use the following notation for $\theta \in (\theta_{i - 1}, \theta_i)$
\[\cspn{i}\!\coloneqq\!\cspn{\theta}, \hfill \eact{i}\!\coloneqq\!\eact{\theta}, \hfill \eres{i}\!\coloneqq\!\eres{\theta}, \hfill \efull{i}\!\coloneqq\!\efull{\theta}, \hfill \xp{i}\!\coloneqq\!\xp{}(\theta), \hfill \lpp{i}{v}\!\coloneqq\! \lp{v}(\theta), \hfill c_{i,v} \!\coloneqq\! c_v(\l_v(\theta)).\]
The inflow at the sink is also constant in a phase. We denote this by the \emph{capacity} $\ca{i}\coloneqq\fp{t}(\l_t(\theta))$ for some $\theta \in (\theta_{i -1}, \theta_i)$ where we use $\fp{t}(\theta) \coloneqq \sum_{vt \in \inset{t}} \fm{vt}(\theta)$. Finally, the derivatives of the waiting times $(q_e(\l_u(\theta)))'$ stay constant within a phase as they are either $\l'_v(\theta) - \l'_u(\theta)$ if $e = uv$ is active or $0$ otherwise. For $\theta \in (\theta_{i - 1}, \theta_i)$ we write $q'_{i, e} \coloneqq (q_e(\l_u(\theta)))'$ and $\qpp{i}{p} \coloneqq \sum_{e\in p} \qpp{i}{e}$ for an $s$-$t$ path $p$.

\section{Lower Bounds on the Price of Anarchy}\label{sec:lower}
We first show in \ref{subsec:unbounded_poa} that the PoA can be unbounded even on very simple graphs (an observation first made in~\cite{sering19}) and that the same is true for graphs with unit capacities. Afterwards, in \ref{subsec:braess_ratio} and \ref{subsec:spillback_faster}, we use similar constructions to investigate the Braess paradox for flows over time with spillback and to show that there exist networks on which Nash flows over time with spillback are faster than their respective counterparts without spillback.

\paragraph{PoA depending on graph structure or capacities.} \label{subsec:unbounded_poa}
Consider the network $\Gamma$ given in Figure \ref{fig_unbounded} and a Nash flow over time $f$ of it. Since in the first phase the shortest path is $(e_1,e_2)$, edge $e_2$ fills up quickly. Once this happens flow already queues up at the end of edge $e_1$, and thus, $e_3$ is never used by $f$. An optimal flow can use $e_3$ and therefore routes flow to the sink much faster, resulting in an unbounded PoA. This construction can easily be generalized to all graphs that have the graph given in Figure \ref{fig_unbounded} as a minor. 

\ifarxiv
The complete calculation for this as well as all skipped proofs of this section can be found in Appendix \ref{app:lower}.
\fi

\begin{restatable}{theorem}{unbounded}\label{thm_unbounded}(cf.~\cite[introductary example]{sering19})
Let $G$ be any graph that has the graph given in Figure \ref{fig_unbounded} as a minor, then there exists a temporal routing game $\Gamma$ on $G$ with
$\mtpoa{\Gamma} \in \Omega(\frac{1}{\num{\text{min}}})$
where $\num{\text{min}} \coloneqq \min_{e\in E} \{\num{e} : \num{e} > 0 \}$.
\end{restatable}
\begin{figure}[t]
\begin{center}
\includegraphics{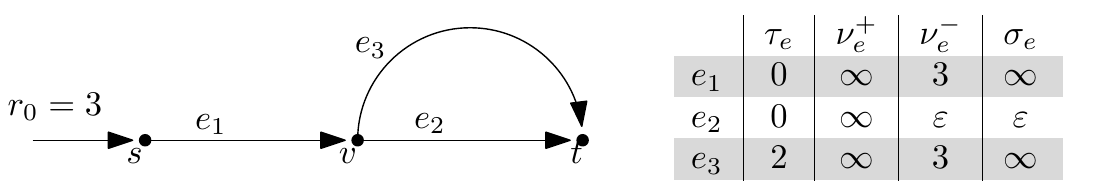}
\caption[Nash flows with spillback have unbounded price of anarchy]{This is a network on which Nash flows over time with spillback have unbounded price of anarchy (see Theorem \ref{thm_unbounded}). A similar example was first given in~\cite[Figure~2]{sering19}.} 
\label{fig_unbounded}
\end{center}
\end{figure}

To avoid the above unboundedness one could ask for the PoA for temporal routing games on graphs with restricted edge capacities. By constrictions of the model we have to set the inflow and storage capacities of all edges $e\in\outset{s}$ to $\nup{e} > r$ and $\sigma_e = \infty$, respectively. We say a network has \emph{unit edge capacities} if for all edges $e\notin\outset{s}$ it holds that $\nup{e} = \num{e} = \sigma_e = 1$ and further for all edges $e\in\outset{s}$ also $\num{e}=1$. Unfortunately, even when restricting to networks with unit edge capacities the PoA is unbounded.

\begin{restatable}{theorem}{capacities}\label{thm_unit_capacities}
There exists a family of networks with unit edge capacities and
$\tau_e \in\{0,1\}$ for all edges
for which the PoA is linear in the number of edges.
\end{restatable}
This can be seen by considering the network given in Figure \ref{fig_unbounded} and exchanging $e_1$ and $e_3$ with bunches of unit-capacity parallel edges and setting $\num{e_2} = \sigma_{e_2} = 1$. If we use enough parallel edges we can generate a similar flow behavior as we encountered when lowering the capacity of edge $e_2$ in the proof of \Cref{thm_unbounded}.

Nevertheless, we show another way of constraining edge capacities to achive an interesting upper bound on the PoA in Section \ref{subsec:saturated}.

\paragraph{Braess ratio.} \label{subsec:braess_ratio}
In his work on selfish routing with static flows~\cite{braess68} Braess showed that there are networks where adding an edge can paradoxically increase congestion leading to a worse equilibrium. In line with the paper of Macko et al.~\cite{macko10} we define the \emph{Braess ratio} for flows over time with spillback as follows.
Let $\Gamma$ be a temporal routing game on a graph $G$ and let $\Gamma(H)$ be the same instance restricted to some subgraph $H\subseteq G$. Then the Braess ratio of $\Gamma$ is
\begin{equation*}
\br{\Gamma} = \max_{H\subseteq G} \frac{\eq(\Gamma)}{\eq(\Gamma(H))}.
\end{equation*}
We say graph $G$ \emph{admits a Braess paradox} if there is a temporal routing game $\Gamma$ on $G$ with $\br{\Gamma} > 1$.
In~\cite{macko10} it is shown that the Braess ratio for flows over time without spillback (for a slightly different cost function instead of the last completion time) is arbitrarily large depending linearly on the number of edges of the underlying graph. The authors furthermore show that a graph $G$ or its transpose (the graph where every edge $uv$ is replaced by the edge $vu$ and $s$ and $t$ are swapped)
admit a Braess paradox if and only if $G$ contains at least one of the following graphs as a topological minor.
\begin{figure}[h]
\begin{center}
\includegraphics{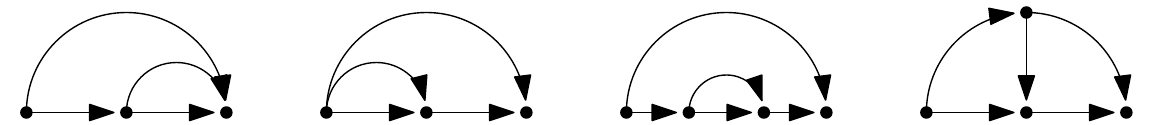}
\end{center}
\end{figure}

When considering flows over time with spillback and the graph in Figure \ref{fig_unbounded} it is easy to see that this graph admits a Braess paradox with arbitrarily large Braess ratio even though it does not have one of the graphs above as a topological minor (and neither does its transpose). 
To see this choose $H$ to be the subgraph where from the graph in Figure \ref{fig_unbounded} we delete edge $e_{2}$.

\begin{corollary}\label{thm_braess}
For any $a \in \R$ there exists a temporal routing game $\Gamma$ on the graph given in Figure \ref{fig_unbounded} such that the Braess ratio satisfies $\br{\Gamma} > a$.
\end{corollary}

\paragraph{Spillback can improve completion time.} \label{subsec:spillback_faster}
The following proposition shows that there are temporal routing games where Nash flows with spillback perform better than Nash flows without spillback. This might at first be surprising, as spillback seems to only be obstructive to routing flow fast. But it is indeed possible to construct networks where spillback leads to shorter completion times. 
In the network depicted in \Cref{fig_fast_spillback} there are two parallel edges, namely $e_3$ and $e_{3^\prime}$, for which it holds that the completion time of a Nash flow is worse if the edges are present compared to the same network without those edges. 
We exploit this in our construction: In the spillback model one of these `bad' edges becomes full nearly instantaneously  yielding the other `bad' edge to never get active. Thus, the spillback Nash flow routes flow only over one of those `bad' edges. Since in the Koch-Skutella model without spillback both of these parallel edges get active at some point, the Nash flow over time here uses both of them resulting in a worse completion time. 
\ifarxiv
For a thorough calculation of the phases of the two Nash flows and the ensuing completion times we refer to the detailed proof in the appendix.
\fi

\begin{restatable}{proposition}{propfastspillback}\label{prop_fast_spillback}
In the network $\Gamma$ given in \Cref{fig_fast_spillback} the completion time of any Nash flow over time with spillback is less than the completion time of the Nash flow over time without spillback on the same network using $\nu_e\coloneqq\min\{\nup{e},\num{e}\}$.
\end{restatable}

\begin{figure}[t]
\begin{center}
\includegraphics{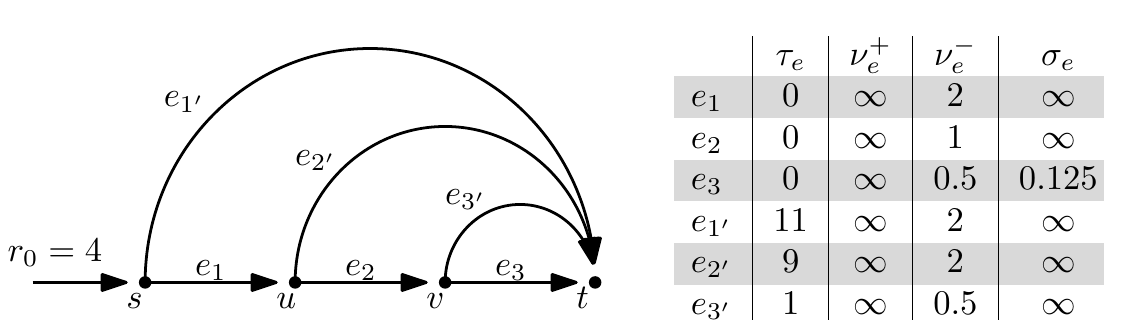}
\caption[Nash flows with spillback can be faster than Nash flows without spillback]{This example shows that Nash flows over time with spillback can be faster than Nash flows over time without spillback, see Proposition \ref{prop_fast_spillback}.} 
\label{fig_fast_spillback}
\end{center}
\end{figure}

\section{Upper Bounds on the Price of Anarchy}\label{sec:upper}
In the following we prove two upper bounds on the price of anarchy. First, 
we show that for a single, fixed network the PoA is bounded by a constant in the long run. After that 
we show that if for a given network we are allowed to decrease the outflow capacities by a certain amount then the PoA only depends on the worst spillback factor of the Nash flows over time.

\subsection{Price of Anarchy for a fixed Network} \label{subsec:poa_for_fixed_networks}
Until now we have studied the PoA depending on the structure of the underlying graph or its capacities. For both questions we constructed 
games satisfying strong constraints that still have unbounded PoA.
Now we are interested in the PoA of a network where every parameter is fixed except for the target amount $M$, i.e., we ask the question of how the PoA behaves in the long run on a single network. 

\begin{restatable}{lemma}{asymptoticOptimal}\label{lem_asymptotic_optimal}
For a temporal routing game $\Gamma$ on a fixed network the completion time of the optimal flow depending on the target amount $M$ is bounded by 
$\opt(M) \in \Theta(M)$.
\end{restatable}
This result is mainly due to the fact that the optimal flow does not build up any queues. Therefore its completion time depends mainly on $M$ and the minimum edge-capacity, which we consider to be fixed. 
\ifarxiv
The full proof as well as all skipped proofs of this section can be found in Appendix \ref{app:upper}.
\fi

For the classification of the asymptotic long term behavior of $\eq(M)$ we use the following auxiliary lemma that gives us a lower bound on the spillback factors of a Nash flow. The lemma follows by an application of~\cite[Lemma~3]{sering19}.

\begin{restatable}{lemma}{lemasymptotic}\label{lem_asymptotic_min_eps}
For a temporal routing game $\Gamma$ there exists an $\epsilon > 0$ such that for any Nash flow over time $f\in\allf{\Gamma}$ the spillback factors satisfy $\min\{c_v(\theta) : v\in V, \theta \in \R_{\geq 0}\} > \epsilon$.
\end{restatable}

To get an asymptotic bound on $\eq(M) = \sup_{f\in\allf{\Gamma}} \comp{f}(M)$ we first argue that seen as a function in $M$, $\comp{f}$ is a piece-wise linear and non-decreasing function. We can then use \Cref{lem_asymptotic_min_eps} to bound its derivative and with that obtain the desired result.

\begin{restatable}{theorem}{asymptoticEQ}\label{thm_asymptotic_eq}
For a temporal routing game $\Gamma$ on a fixed network the completion time of any Nash flow over time $f\in\allf{\Gamma}$ is bounded by $\comp{f}(M) \in \Theta(M)$.
\end{restatable}

We can now use Lemma \ref{lem_asymptotic_optimal} and Theorem \ref{thm_asymptotic_eq} and the fact that $\mtpoa{M} = \frac{\eq(M)}{\opt(M)}$ to bound the PoA for a fixed network. In order to do so we consider the PoA as a function of the target flow amount $M$. 

\begin{theorem}\label{thm_asymptotic_constant}
For a temporal game $\Gamma$ on a fixed network, i.e. when treating everything except the amount of flow $M$ as a constant, the price of anarchy is bounded by a constant, $\mtpoa{M} \in \Theta(1)$.
\end{theorem}


\subsection{Bound on the Price of Anarchy for Saturated Graphs}\label{subsec:saturated}
In this section we focus on networks with an additional constraint on the edge capacities. Given a game $\Gamma$ we know that the quickest flow of $\Gamma$ is also a temporally repeated flow, i.e., it has an underlying static flow $\fh$. We say that $\fh$ \emph{saturates every edge} of the given graph if for each edge the outflow capacity is exhausted by $\fh$, i.e., for each $e\in E$ we have $\num{e} = \fh_e$ and additionally it holds that $|\fh|=\sum_{sv\in\outset{s}} \fh_{sv} = r_0$. We call the underlying graph of such a game a \emph{saturated graph}.
Even though restricting attention to saturated graphs may seem harsh, note, that every network can be made saturated by lowering the edge capacities. This can be imagined to be done by a system operator in a Stackelberg strategy-like scenario~\cite{stackelberg34} and is applicable in many real-world examples. For one, streets can be narrowed down by a city administration in practice.

For temporal routing games on saturated graphs we will show that the PoA can be bounded by a value that is only dependent on the worst spillback factor of all Nash flows over time. In order to do that we adapt the idea of the proofs given by Bhaskar et al.~\cite{bhaskar15} for the Koch-Skutella model to the spillback model. 
Note, however, that the proofs given in~\cite{bhaskar15} implicitly assume only finitely many phases, which has not been proven for any of the two models.
Our generalization also holds for the case of an infinite number of phases in both models.\footnote{Note, that in~\cite{correa19} an even more general result is shown for the Koch-Skutella model.}

In principle the proof works as follows. For a given game $\Gamma$ the relation of the completion time of any Nash flow over time of $\Gamma$ to the optimal completion time can be determined by examining the capacity of the current shortest path network and the derivatives of the waiting times for a single phase of the Nash flow. One can then bound the derivatives of the waiting times and use the fact that the PoA is the maximum over the relation of the optimal completion time to the completion times of all Nash flows. This achieves the desired bound.

\paragraph{Bound on the derivatives of the waiting times.}
We start by proving a relation between the derivative of the label-function at the sink and the inflow into the sink. Our proof of this result uses a different idea than the one given in~\cite{bhaskar15} and is considerably shorter.
\begin{lemma}\label{III.L2}(cf.\cite[Lemma~15]{bhaskar15})
Let $\Gamma$ be a temporal routing game and let $f\in\allf{\Gamma}$ with corresponding labels $\l$. Then for any $\theta\leq \frac{M}{r_0}$ we have
\begin{equation*}
\lp{t}(\theta) = \frac{r_0}{\fp{t}(\l_t(\theta))}.
\end{equation*}
\end{lemma}

\begin{proof}
Let $\theta\leq\frac{M}{r_0}$ be arbitrary. Using that $\xp{}(\theta)$ is a static $s$-$t$ flow of value $r_0$ and $\xp{vt}(\theta) = \fm{vt}(\l_t(\theta)) \cdot \lp{t}(\theta)$ from Equation \eqref{eq:thin_flows}  we obtain
\begin{equation*}
r_{0} = \sum_{vt\in\inset{t}} \xp{vt}(\theta) = \sum_{vt\in\inset{t}} \fm{vt} (\l_t(\theta)) \cdot \lp{t}(\theta) = \lp{t} (\theta) \cdot \fp{t}(\l_t(\theta)).
\end{equation*}
Since $\fp{t}(\l_t(\theta))>0$ for all $\theta$, rearranging terms give the desired result.
\end{proof}

We now proceed with a path-wise bound on the derivatives of the waiting times $\qpp{i}{p}$ for a single phase of the Nash flow over time $i$ using the capacities~$\ca{i}$.

\begin{restatable}{lemma}{IIILfour}\label{III.L4}(cf.\cite[Lemma~18]{bhaskar15})
Let $\Gamma$ be a temporal routing game where the static flow underlying the quickest flow saturates every edge and let $f\in\allf{\Gamma}$. For any $s$-$t$ path $p$, the travel time is bounded by
\begin{equation*}
\tau_p \geq \ome{r} - \sum_{i\in\phases} ( 1 + \qpp{i}{p}) \cdot \frac{\ca{i}}{r_0} \cdot (\ome{i} - \ome{i-1}).
\end{equation*}
\end{restatable}

In the proof we first establish a dependence of the length of a phase as it is experienced at the source and at the sink, respectively. Then we express $\tau_{p}$ in terms of the label functions $\l$ and the waiting times $q$ and their derivatives. The result then follows from applying Lemma \ref{III.L2}.

\paragraph{Relation of the completion times of Nash flow and quickest flow.}
The following lemma enables us to give a first relation of the completion times of the optimal quickest flow and a Nash flow over time.

\begin{restatable}{lemma}{IIILfive}\label{III.L5}(cf.~\cite[Lemma~19]{bhaskar15})
Let $\Gamma$ be a temporal routing game where the static flow $\fh$ underlying the quickest flow saturates every edge and let $f\in\allf{\Gamma}$. Then, the completion time $\opt$ of the optimal flow and the completion time $\comp{f}$ of the Nash flow $f$ are related as
\begin{equation*}
r_0 \cdot \opt = \sum_{p\in \Paths_{s,t}} \fh_p \tau_p + \sum_{i\in\phases} \ca{i} \cdot (\ome{i} - \ome{i-1}),
\end{equation*}
where $\Paths_{s,t}$ is the set of all simple $s$-$t$ paths in $G$ and $\ome{r} = \comp{f}$.
\end{restatable}
The proof idea is to compare the arrival rates of both flows at the sink  $t$ where we use a flow decomposition along paths for the optimal flow and a decomposition by phases for the Nash flow over time.

By combining the previous two lemmas we can now derive a lower bound on the inverse of the PoA that we will afterwards use to achieve an upper bound on the actual PoA. But in order to proof that we first need the following.

\begin{restatable}{lemma}{IIILseven}\label{III.L7}
Let $\lambda_i \coloneqq \frac{\ca{i}}{r_0} \cdot \sum_{p\in\Paths_{s,t}} \fh_p \qpp{i}{p}$ for each phase $i\in\phases$. Then,
\begin{equation*}
\sum_{i\in\phases} \lambda_i \cdot (\ome{i}-\ome{i-1}) \leq (\ome{r}-\ome{0}) \cdot \sup_{i\in\phases} \lambda_i.
\end{equation*}
\end{restatable}

In the proof we first establish that the set $\{\lambda_i : i\in\phases\}$ is bounded and then use this and the telescoping principle to bound the left hand side.

The next lemma establishes the aforementioned bound on the inverse of the PoA. It is in this proof that the number of $\alpha$-extension phases comes into play. If we assume that the supremum in the statement of Lemma \ref{III.L7} is attained by some phase $i\in\phases$, which is in particular true if there are only finitely many phases, then we can prove Lemma \ref{III.C6} without the $\epsilon$ error and the proofs go through similar to~\cite{bhaskar15}. But since it is still an open problem whether the number of those phases is always finite (in the Koch-Skutella model as well as the spillback model), we prove it here for the case of infinitely many $\alpha$-extension phases.

\begin{restatable}{lemma}{IIICsix}\label{III.C6}
Let $\Gamma$ be a temporal routing game where the static flow $\fh$ underlying the quickest flow saturates every edge and let $f\in\allf{\Gamma}$. Then for every $\epsilon > 0$ there exists a phase $i$ of $f$ such that
\begin{equation*}
\frac{\opt}{\comp{f}} + \epsilon \geq 1 - \frac{\ca{i}}{{r_0}^2} \sum_{e\in E} \num{e} \qpp{i}{e}.
\end{equation*}
\end{restatable}

The proof idea is to sum $\fh_p \tau_p$ over all paths $p\in\Paths_{s,t}$ and using Lemma  \ref{III.L4} to bound this from below. Afterwards, we use \Cref{III.L5,III.L7} to obtain a lower bound on $\frac{\opt}{\comp{f}}$ in terms of a supremum of the capacities and derivatives of the queuing delay over all phases. Since we do not know whether this supremum is attained we have to inject the $\epsilon$ error and after rearranging terms we obtain the desired result.

\paragraph{Upper bound for saturated graphs.}
We can now turn the lower bound in Lemma \ref{III.C6} into an upper bound on the price of anarchy by proving a bound on the sum of the right-hand side of the expression given in Lemma \ref{III.C6}. Here for the first time the spillback factors of the Nash flow over time play an important role. 

\begin{restatable}{lemma}{IIILeight}\label{III.L8}
Let $\Gamma$ be a temporal routing game and $f\in\allf{\Gamma}$. In any phase $i$ of $f$ where $\frac{r_0}{\ca{i}}\geq 1$ we have
\begin{equation*}
\sum_{e\in E} \num{e} \qpp{i}{e} \leq \frac{r_0}{\minc{f}{i}} \ln\left(\frac{r_0}{\ca{i}}\right),
\end{equation*}
where $\minc{f}{i} \coloneqq \min\{c_v(\theta) : v\in V, \theta \in (\theta_{i-1}, \theta_i)\}$ is the minimal $c_v$ of $f$ in phase $i$.
\end{restatable}

The proof utilizes~\cite[Claim 12]{correa19} and follows the line of argumentation in~\cite{bhaskar15} but incorporates the added complexity of the spillback model.
We obtain that $c_v \num{e} \qp{e} = \xp{e} \cdot (1- \frac{\lp{u}}{\lp{v}})$ for every edge $e=uv$ and then sum this expression over all edges in the graph. Rearranging and plugging in the above expression then yields the desired result.

We can now obtain the desired upper bound on the price of anarchy.

\begin{theorem}\label{III.T9}
Let $\Gamma$ be a temporal routing game where the static flow $\fh$ underlying the quickest flow saturates every edge of the graph.
If the minimal spillback factor satisfies $\minc{}{} \coloneqq \min_{f\in\allf{\Gamma}}\min\{c_v(\theta) : v\in V, \theta \in \R_{\geq 0}\} > \frac{1}{e}$, then the price of anarchy is bounded by $\frac{\eq}{\opt} \leq \frac{\minc{}{} e}{\minc{}{} e - 1}.$
\end{theorem}

\begin{proof}
For any $f\in\allf{\Gamma}$ with completion time $\comp{f}$ we know that $\fp{t}(\theta) = \sum_{vt\in\inset{t}} \fm{vt}(\theta) \leq  r_0$ for all $\theta \in \R_{\geq 0}$ since we only consider saturated graphs. Thus, we have $\frac{r_0}{\ca{i}} \geq 1$ in all phases of $f$. From \Cref{III.C6,III.L8} we obtain that for every $\epsilon > 0$ there exists a phase $i$ of $f$ such that
\begin{equation*}
\frac{\opt}{\comp{f}} + \epsilon \geq 1 - \frac{\ca{i}}{{r_0}^2} \sum_{e\in E} \num{e} \qpp{i}{e} \geq 1 - \frac{\ca{i}}{{r_0}^2} \frac{r_0}{\minc{f}{i}} \cdot \ln\left(\frac{r_0}{\ca{i}}\right) = 1 - \frac{a_{i}}{\minc{}{}} \cdot \ln\left(\frac{1}{a_{i}}\right),
\end{equation*}
where $\minc{}{} \coloneqq \min_{f\in\allf{\Gamma}}\min\{c_v(\theta) : v\in V, \theta \in \R_{\geq 0}\} \leq \minc{f}{i}$ and $a_{i}\coloneqq\frac{\ca{i}}{r_0}$.

Simple calculus shows that the term $\frac{a_{i}}{\minc{}{}} \cdot \ln\left(\frac{1}{a_{i}}\right)$ is maximized for $a_{i} = \frac{1}{e}$.
Using the above inequality, derived from some phase $i$, for any $\epsilon > 0$ we obtain
\begin{equation*}
\frac{\opt}{\comp{f}} + \epsilon \geq 1 - \frac{1}{\minc{}{} e} = \frac{\minc{}{} e - 1}{\minc{}{} e}.
\end{equation*}
Since by assumption we have $\minc{}{} > \frac{1}{e}$ we can take the inverse of the inequality to obtain $\frac{\comp{f}}{\opt} \leq \frac{\minc{}{} e}{\minc{}{} e - 1}.$
We finish by noting that $\eq = \sup_{f\in\allf{\Gamma}} \comp{f}$.
\end{proof}

\section{Conclusions} \label{sec:conclusion}
Our work shows that the PoA is highly dependent on spillback effects. Although, even in restricted network classes the completion times of dynamic equilibria can be arbitrarily bad compared to a quickest flow, the PoA can still be bounded in terms of the spillback factors under some constraints on the edge capacities. Transferred to real-world traffic this means the interplay between selfish traffic users is critical in particular in high congested areas.

Even though we give a substantial analysis of the PoA in the flow over time model with spillback, there are still some open problems remaining. Is the bound we establish in \Cref{III.T9} tight? Are there any bounds in the case of $\minc{}{} \leq \frac{1}{e}$ or is it possible to enforce $\minc{}{} > \frac{1}{e}$ through some Stackelberg-like strategy? Do the results of the recent work of Correa et al.~\cite{correa19} also transfer to the spillback setting? On the more applied side of the research it would also be very interesting to algorithmically identify street segments (edges) which are especially vulnerable for spillback. In the long run this could help road administrations to decide which roads should be expanded (increasing the storage capacity) or which roads should be narrowed or closed (due to the Braess effect).


\bibliographystyle{splncs04}
\bibliography{bibfile}

\newpage

\ifarxiv
\appendix
\section{Proofs of Section \ref{sec:lower}}\label{app:lower}
Here we provide the full proofs that were skipped in Section \ref{sec:lower}. 

\unbounded*
\begin{proof}
Consider the network $\Gamma$ given in Figure \ref{fig_unbounded} and let $f$ be a Nash flow over time of it. We see that edge $e_2$ fills up at time $\frac{\epsilon}{3-\epsilon}$ at which the travel time on the straight lower path is 1. But since edge $e_3$ has a free flow transit time of $\tau_3 = 2$, $e_3$ never enters the shortest path network.
Since $\fp{t}(\theta) = \epsilon$ for all times $\theta \in\R_{\geq 0}$, the completion time of Nash flow $f$ depending on the target flow amount $M$ is $\comp{f}(M) = \frac{M}{\epsilon}$.

On the other hand the quickest flow $g$ of the given network routes $\epsilon$ units of flow over $e_2$ and (supposing that $M\geq 2\epsilon$) routes the remaining $3-\epsilon$ over $e_3$ resulting in a completion time of $\opt(M) = 2 + \frac{M-2\epsilon}{3}$.
For any temporal routing game $\Gamma$ on the given network $\Gamma$ with fixed target amount $M\in \R_{\geq 0}$ we thus have
\begin{equation*}
\mtpoa{\Gamma} \geq \frac{\comp{f}(M)}{\opt(M)} = \frac{\frac{M}{\epsilon}}{2 + \frac{M-2\epsilon}{3}} \rightarrow \infty \text{ ~~~for } \epsilon \rightarrow 0.
\end{equation*}
\end{proof}

\capacities*
\begin{proof}
Let $k\in \N$ and consider the network in Figure \ref{fig_unbounded} where we exchange $e_1$ and $e_3$ by bunches of parallel edges $E_1$ and $E_3$, respectively, with $r_0=|E_1|=k$ and $|E_3|=k-1$. Furthermore, we set all edge capacities to $1$ and set $\tau_e = 0$ for $e\in E_1$ and $\tau_e = 1$ for $e\in E_3$. Similar to the proof of \Cref{thm_unbounded}, the Nash flow over time $f$ will not use any of the edges $e\in E_3$ but the quickest flow~$g$ will. Thus, for all times $\theta\in\R_{\geq 0}$ we have $\fp{t}(\theta)=1$ but for the quickest flow $g$ it holds that $\gp{t}(\theta)=1$ only if $\theta < 2$ and otherwise $\gp{t}(\theta) = k$. If we set $M=k+2$ we get $\mtpoa{\Gamma} \geq \frac{\comp{f}}{\opt} = \frac{M}{2 + \frac{M-2}{k}}=\frac{k+2}{3}$, i.e., $\mtpoa{\Gamma}\in\Omega(k)$.
\end{proof}

\begin{figure}[t]
\begin{center}
\includegraphics[scale=0.9]{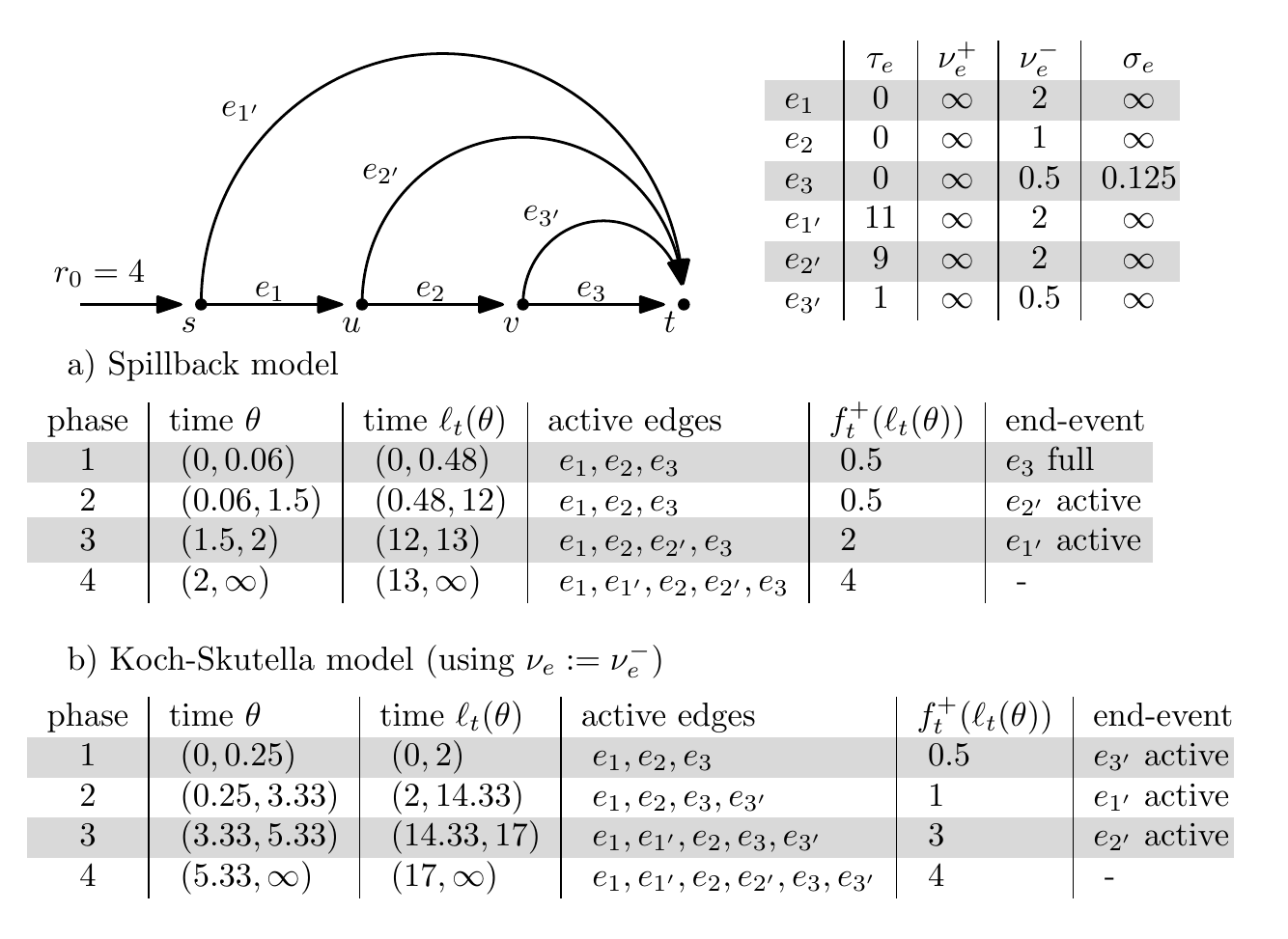}
\caption{This is the same example as in Figure \ref{fig_fast_spillback} that shows that Nash flows over time with spillback can be faster than Nash flows over time without spillback, see Proposition \ref{prop_fast_spillback}. Additionally, here we give the Nash flows of the two models of flows over time.} 
\label{fig_fast_spillback_table}
\end{center}
\end{figure}
\propfastspillback*
\begin{proof}
Consider the network and tables given in Figure \ref{fig_fast_spillback_table}. We show that the statement is true for all target amounts $M$ that are greater than some threshold. In the spillback model edge $e_3$ gets full very quickly which prevents the Nash flow from using $e_{3^\prime}$ (see Table a)). In the final phase of the spillback Nash flow, i.e.\ for all $\theta>2$, we have $T_{s,t}(\theta) = \theta +11$. That means from time $\theta=2$ on it takes the Nash flows 11 time units to route flow from the source to the sink.
On the other hand the Nash flow of the Koch-Skutella model uses both of the 'bad' edges $e_3$ and $e_{3^\prime}$ and thus in its final phase for all $\theta>5.33$ we have $T_{s,t}(\theta)=\theta+11.66$ (see Table b)). That means the Nash flow without spillback needs longer to route the flow from the source to the sink. 
For a temporal routing game on this network it thus holds that for all $M>\frac{40}{3}$ the completion time of the above Nash flow with spillback is less than the completion time of the Nash flow of the model without spillback. 
\end{proof}

\section{Proofs of Section \ref{sec:upper}}\label{app:upper}
Here we provide the full proofs that were skipped in Section \ref{sec:upper}. 

\asymptoticOptimal*
\begin{proof}
Let $g$ be an optimal flow and $p$ be the shortest $s$-$t$ path in the network without any queues. We clearly have for all $M$ that $\opt(M) \geq \frac{M}{r_0}$ as $g$ can not route flow before it is present at the source. On the other hand, it is known that the optimal flow does not build up any queues~\cite{koch11} from which we can deduce that $\opt(M) \leq \tau_{p} + \frac{M}{\nu_{p}},$ for all $M$, where $\nu_p$ is the minimum of all in- and outflow capacities along $p$.
\end{proof}

\lemasymptotic*
\begin{proof}
Let $f\in\allf{\Gamma}$ be a Nash flow over time of the network $\Gamma$ and let $c_v(\theta) < 1$ for some vertex $v\in V$ at time $\theta\in \R_{\geq 0}$. Then there exists at least one throttled incoming edge $e=uv\in\inset{v}$. By definition $c_v(\theta)$ is the maximum value to fulfil the fair allocation condition $\fm{e}(\theta) = \min\{\bm{e}(\theta), \num{e} c_v(\theta)\}$. This yields
\begin{equation}\label{lem_asymptotic_min_eps.1}
c_v(\theta) \geq \frac{\fm{e}(\theta)}{\num{e}}.
\end{equation}
To bound $\fm{e}(\theta)$ we follow the idea of the proof of~\cite[Lemma~3]{sering19}. For that let $\nu_{\text{min}} \coloneqq \min(\{\nup{e},\num{e} : e\in E\}\cup \{1\})$ and $\num{\Sigma} \coloneqq \max\{\sum_{e\in E} \num{e}, 1\}$.
Since $e=uv$ is throttled we know from the no slack condition that there is an edge $e_1=vw\in\outset{v}$ with $\fp{e_1}(\theta) = \bp{e_1}(\theta)$. If $e_1$ is full and throttled, we continue with an edge $e_2\in\outset{w}$ for which it holds that $\fp{e_2}(\theta) = \bp{e_2}(\theta)$. This edge again exists because of the no slack condition. We carry on this procedure until we find an edge $e_k$ that is not throttled or not full. Since we know that the set of full edges is acyclic by the no deadlock condition the procedure terminates, and furthermore, we know that $\fp{e_k}(\theta) = \bp{e_k}(\theta)\geq \min\set{\nup{e_k},\num{e_k}}$ since it is either not full ($\bp{e_k} = \nup{e_k}$) or full but not throttled ($\bp{e_k} = \min\set{\nup{e_k},f_e^-(\theta)}$ and $f_e^-(\theta) = b_e^-(\theta) = \nup{e_k}$). If we consider two consecutive edges $e_{i-1}\in\inset{v^\prime}$ and $e_i\in\outset{v^\prime}$ of the above procedure and use that $e_{i-1}$ is throttled we get 
\begin{equation}\label{lem_asymptotic_min_eps.2}
\fm{e_{i-1}}(\theta) = c_v(\theta) \cdot \num{e_{i-1}} \geq \frac{\sum_{e^\prime \in\outset{v^\prime}} \fp{e^\prime}(\theta)}{\sum_{e^\prime \in\inset{v^\prime}} \cdot \num{e^\prime}} \cdot \nu_{\text{min}} \geq \frac{\fp{e_i}(\theta)}{\num{\Sigma}} \cdot \nu_{\text{min}},
\end{equation}
where in the first inequality we used the definitions of $c_v$ and $\nu_{\text{min}}$ and in the second the definition of $\num{\Sigma}$ and the fact that $e_i\in\outset{v^\prime}$. By construction of the procedure we know that for all $i\leq k$ edge $e_i$ is full with saturated inflow capacity and thus it holds that $\fp{e_i}(\theta) = \bp{e_i}(\theta) = \min\{\nup{e_k},\num{e_k}\}$. Now we can recursively apply the above Equation \eqref{lem_asymptotic_min_eps.2} on all edges $e_i, i\leq k$ which gives us the desired bound on the outflow rate of edge $e$:
\begin{equation*}
\fm{e} \geq \left(\frac{\nu_{\text{min}}}{\num{\Sigma}}\right)^k \cdot \nu_{\text{min}} \geq \left(\frac{\nu_{\text{min}}}{\num{\Sigma}}\right)^{|E|} \cdot \nu_{\text{min}} \coloneqq \epsilon^\prime,
\end{equation*}
where we used that $k\leq |E|$.

Plugging this into Inequality \eqref{lem_asymptotic_min_eps.1} provides 
\begin{equation*}
c_v(\theta) \geq \frac{\fm{e}(\theta)}{\num{e}} \geq \frac{\epsilon^\prime}{\num{e}}.
\end{equation*}
Let $\num{\text{max}}\coloneqq \max_{e\in E} \num{e}$. Then for any vertex $v\in V$ at any time $\theta\in\R_{\geq 0}$ it holds that $c_v(\theta) \geq \frac{\epsilon^\prime}{\num{\text{max}}} =: \epsilon,$ where $\epsilon = \left(\frac{\nu_{\text{min}}}{\num{\Sigma}}\right)^{|E|} \cdot \frac{\nu_{\text{min}}}{\num{\text{max}}}$ only depends on the network $\Gamma$ and not on the specific Nash flow or the target amount.
\end{proof}

\asymptoticEQ*
\begin{proof}
Let $f\in\allf{\Gamma}$ be a Nash flow over time on the network $\Gamma$, then $\comp{f}(M) = \l_t\left(\frac{M}{r_0}\right)$ is the time by which $f$ has routed $M$ units of flow to the sink $t$. 
We know from the construction of the $\alpha$-extension~\cite[Section 5]{sering19} that $\lp{t}(\theta)$ is non-negative and constant per phase of $f$. Therefore, $\comp{f}(M): \R_{\geq 0} \rightarrow \R_{\geq 0}$ with $M\mapsto \l_t\left(\frac{M}{r_0}\right)$ is a piece-wise linear function where each linear segment is ascending. 
From the definition of thin flows it follows that
\begin{equation*}
\sup_{\theta\in\R_{\geq 0}} \lp{t}(\theta) \leq \max\left\lbrace\left. \frac{\xp{e}(\theta)}{c_v(\theta)\cdot \num{e}} \right. : v\in V, e\in\inset{v}\right\rbrace.
\end{equation*}
Noting that $\xp{}(\theta)$ is a static flow of value~$r_0$ and using Lemma \ref{lem_asymptotic_min_eps} we obtain $\lp{t}(\theta) \leq \frac{r_0}{\epsilon \cdot \num{\text{min}}}$ for all times $\theta \in \R_{\geq 0}$, where $\num{\text{min}} \coloneqq \min_{e\in E} \num{e}$ is given by the parameters of the network.
Thus, $\comp{f}(M)$ consists of linear segments of slope at most $\frac{r_0}{\epsilon \cdot \num{\text{min}}}$ with $\comp{f}(0)=\l_t(0)$, and therefore, it holds for all $M\in\R_{\geq 0}$ that 
\begin{equation*}
\comp{f}(M) \leq \frac{M}{r_0}\cdot \frac{r_0}{\epsilon \cdot \num{\text{min}}} + \l_t(0) = M \cdot \frac{1}{\epsilon \cdot \num{\text{min}}} +\l_t(0) \in \bigO(M).
\end{equation*}
Conversely, as in the proof of \Cref{lem_asymptotic_optimal} we have $\comp{f}(M) \geq \frac{M}{r_0}$. 
\end{proof}


\IIILfour*
In order to prove this we first establish the following lemma on the dependence of the length of a phase as it is experienced at the source and at the sink, respectively. We deferred the proof to the appendix as it is equivalent to the one given in~\cite{bhaskar15}.

\begin{lemma}\label{III.C3}(cf.\cite[Corollary~17]{bhaskar15})
Let $\Gamma=\trg$ be a temporal routing game where the static flow underlying the quickest flow saturates every edge and let $f\in\allf{\Gamma}$. For any phase $i\in\phases$ with boundary points $\theta_{i-1}$ and $\theta_i$ it holds that
\begin{equation*}
\ome{i} - \ome{i-1} = \frac{r_0}{\ca{i}} (\theta_i - \theta_{i-1}).
\end{equation*}
\end{lemma}
\begin{proof}[Proof of Lemma \ref{III.C3}]
By definition we have $\ome{i} - \ome{i-1} = \l_t(\theta_i) - \l_t(\theta_{i-1})$ and $\fp{t}(\l_t(\theta)) = \ca{i}$ for some~$\theta$ in phase $i$. Thus, by definition of the integral
\begin{align*}
\ome{i} - \ome{i-1} &= \int^{\theta_i}_{\theta_{i-1}} \lp{t}(\phi)\text{d}\phi \\
&= \int^{\theta_i}_{\theta_{i-1}} \frac{r_0}{\ca{i}}\text{d}\phi \\
&= \frac{r_0}{\ca{i}} (\theta_i - \theta_{i-1}),
\end{align*}
where we used Lemma \ref{III.L2} for the second equality.
\end{proof}

\begin{proof}[Proof of Lemma \ref{III.L4}]
By the definition of the shortest path network, for any $\theta \in \mathbb{R}$ and any $e=uv\in G_\theta$ for the transit time we can write $\tau_e=\l_w(\theta)-\l_v(\theta)-q_e(\l_v(\theta))$. Furthermore, by the definition of derivatives of piece-wise linear functions we can express the label function for any $v\in V$ and the waiting time of any edge $e=uv$ at time $\theta_r$ in terms of their derivatives as $\l_v(\theta_r)=\sum_{i\in\phases} \lpp{i}{v}\cdot(\theta_i - \theta_{i-1}) + \l_v(\theta_0)$ and $q_e(\l_v(\theta_r))=\sum_{i\in\phases} \qpp{i}{e} \cdot(\theta_i - \theta_{i-1})$, where we used that $q_e(\theta_0)=0$ for all edges, as at time $\theta_0 = 0$ there are no queues.
Hence, we conclude for any $e=vw\in E_{\theta_r}$ that
\begin{equation*}
\tau_e = \l_w(\theta_0)-\l_v(\theta_0) + \sum_{i\in\phases} (\lpp{i}{w} - \lpp{i}{v} - \qpp{i}{e})\cdot(\theta_i - \theta_{i-1}).
\end{equation*}
For any $e=vw\notin E_{\theta_r}$ we have by~\cite[Lemma~7]{sering19} that
\begin{equation*}
\tau_e \geq \l_w(\theta_r) - \l_v(\theta_r) = \l_w(\theta_0)-\l_v(\theta_0) + \sum_{i\in\phases} (\lpp{i}{w} - \lpp{i}{v} - \qpp{i}{e})\cdot(\theta_i - \theta_{i-1}).
\end{equation*}
Summing over all edges in path $p$ and using the identities $\l_s(\theta_0)=\theta_0=0$ and $\lpp{i}{s}=1$ yields
\begin{align*}
\tau_p = \sum_{e\in p} \tau_e &\geq \sum_{e\in p} \l_w(0)-\l_v(0) + \sum_{i\in\phases} (\lpp{i}{w} - \lpp{i}{v} - \qpp{i}{e})\cdot(\theta_i - \theta_{i-1}) \\
&= \l_t(0) - \l_s(0) + \sum_{i\in\phases} (\lpp{i}{t} - \lpp{i}{s} - \qpp{i}{p})\cdot(\theta_i - \theta_{i-1}) \\
&= \ome{0} + \sum_{i\in\phases} (\frac{r_0}{\fp{t}(l_t(\phi_i))} - 1 - \qpp{i}{p})\cdot(\theta_i - \theta_{i-1})
\end{align*}
for some time $\phi_i$ in phase $i$, where for the last equality we used Lemma \ref{III.L2}. Using $\ca{i} = \fp{t}(\l_t(\phi_i))$ and Lemma \ref{III.C3} we can conclude
\begin{align*}
\tau_p &\geq \ome{0} + \sum_{i\in\phases} (\frac{r_0}{\ca{i}} - 1 - \qpp{i}{p})\frac{\ca{i}}{r_0}(\ome{i} - \ome{i-1}) \\
&= \ome{0} + \ome{r} - \ome{0} - \sum_{i\in\phases} (1 + \qpp{i}{p})\frac{\ca{i}}{r_0}(\ome{i} - \ome{i-1}) \\
&= \ome{r} - \sum_{i\in\phases} (1 + \qpp{i}{p})\frac{\ca{i}}{r_0}(\ome{i} - \ome{i-1}).
\end{align*}
\end{proof}

\IIILfive*
\begin{proof}
We compare the arrival rates of both flows at the sink  $t$. 
First, consider the rate of arrival at the sink $\gp{t}$ of the optimal flow. Note that we can always suppose an optimal flow to never build up any queues. Therefore, for any simple $s$-$t$ paths $p\in\Paths_{s,t}$ after time $\tau_p$ the flow on $p$ arrives with rate $\fh_p$ at $t$, i.e.\ $\gp{t}(\theta) = \sum_{p\in\Paths_{s,t}(\theta)} \fh_p$, where $\Paths_{s,t}(\theta)$ is the set of all simple $s$-$t$ paths with $\tau_p \leq \theta$. The total flow arriving at the sink by time $\theta$ is the area under that curve up to time $\theta$. Thus, the total flow $M$ that $g$ routes to $t$ can be expressed as
\begin{equation}\label{III.L5.1}
M = r_0\cdot\opt - \sum_{p\in\Paths_{s,t}} \fh_p\tau_p.
\end{equation}
We now perform a similar calculation for a Nash flow over time $f\in\allf{\Gamma}$.
We again depict $M$ by the inflow rate at the source over several intervals, using the phases of $f$ instead of the path length. The rate with which flow enters the sink in phase $i$ of the dynamic equilibrium $f$ is $\ca{i}$ and the length of that phase experienced at the sink is $\ome{i}-\ome{i-1}$. Thus we can obtain 
\[M = \sum_{i\in\phases} \ca{i}(\ome{i}-\ome{i-1})\]
where $\comp{f}$ is implicitly present as $\ome{r}$.
Equating this with the expression for $M$ with respect to the optimal flow given in \eqref{III.L5.1} yields the desired
\begin{equation*}
r_0 \opt = \sum_{p\in \Paths_{s,t}} \fh_p \tau_p + \sum_{i\in\phases} \ca{i}(\ome{i} - \ome{i-1}). \qedhere
\end{equation*}
\end{proof}

\IIILseven*
\begin{proof}
First note, that $(\ome{i})_{i\in\phases}$ is a strict monotonically increasing sequence that is bounded by $\comp{f} = \l_t(\frac{M}{r_0})=\ome{r}\in\phases$. Furthermore, $\sup_{i\in\phases} \ca{i} = \sup_{i\in\phases} \fp{t}(\l_t(\phi_i)) \leq \sum_{e\in\inset{t}} \num{e}$ for any point in time $\phi_i\in (\theta_{i-1},\theta_i)$ and $\fh_p \leq r_0$. We also know that $\qpp{i}{p} \leq \max_{e\in E, \theta \in \R_{\geq 0}} \frac{M}{\fm{e}(\theta)} \leq \max_{e\in E, \theta \in \R_{\geq 0}} \frac{M}{c_v(\theta)\cdot \nup{e}}$ which is bounded by Lemma \ref{lem_asymptotic_min_eps}. This together yields that the set of $\{\lambda_i = \frac{\ca{i}}{r_0} \cdot \sum_{p\in\Paths_{s,t}} \fh_p \qpp{i}{p} : i\in\phases\}$ is bounded.
We thus can write
\begin{align*}
\sum_{i\in\phases} \lambda_i \cdot (\ome{i} - \ome{i-1} ) &\leq \sum_{i\in\phases} \sup_{i\in\phases} \lambda_i \cdot (\ome{i} - \ome{i-1} ) \\
&= \sup_{i\in\phases}  \sum_{i\in\phases} \lambda_i \cdot (\ome{i} - \ome{i-1} ) \\
&= \sup_{i\in\phases}  (\ome{r} - \ome{0}).
\end{align*}
Since we know that $\sup_{i\in\phases} \ome{i} = \ome{r} = \comp{f}$ we can use the telescoping principle for the infinite sum in the last equality and cancel out all intermediate $\ome{i}$ for $i\in\phases\setminus \{0,r\}$.
\end{proof}

\IIICsix*
\begin{proof}
Let $\Paths_{s,t}$ be the set of all simple $s$-$t$ paths in $G$. Summing $\fh_p \tau_p$ over all paths $p\in\Paths_{s,t}$ and using Lemma  \ref{III.L4} yields
\begin{align*}
\sum_{p\in\Paths_{s,t}} \fh_p \tau_p &\geq \sum_{p\in\Paths_{s,t}} \fh_p \ome{r} - \sum_{p\in\Paths_{s,t}}\left[ \fh_p \sum_{i\in\phases} (1+\qpp{i}{p}) \frac{\ca{i}}{r_0} (\ome{i} - \ome{i-1})\right] \\
&= \ome{r} \sum_{p\in\Paths_{s,t}} \fh_p - \sum_{i\in\phases} \left[\frac{\ca{i}}{r_0} (\ome{i} - \ome{i-1}) \sum_{p\in\Paths_{s,t}} \fh_p(1 + \qpp{i}{p})\right] \\
&= \ome{r} r_0 - \sum_{i\in\phases} \left(\ca{i} + \frac{\ca{i}}{r_0} \sum_{p\in\Paths_{s,t}} \fh_p \qpp{i}{p}\right) (\ome{i} - \ome{i-1}),
\end{align*}
where in the last equality we used the fact that $\fh$ saturates every edge, especially that $\sum_{p\in\Paths_{s,t}} \fh_p = r_0$. Plugging this into the result of Lemma \ref{III.L5} yields
\begin{align*}
r_0 \cdot \opt &\geq \sum_{i\in\phases} \ca{i} (\ome{i} - \ome{i-1}) + \ome{r} r_0 \\
&\phantom{\geq \sum_{i\in\phases} \ca{i}} - \sum_{i\in\phases} \left(\ca{i} + \frac{\ca{i}}{r_0} \sum_{p\in\Paths_{s,t}}\fh_p \qpp{i}{p}\right)(\ome{i} - \ome{i-1}) \\
&= \ome{r} r_0 - \sum_{i\in\phases} \left(\frac{\ca{i}}{r_0}\sum_{p\in\Paths_{s,t}} \fh_p\qpp{i}{p}\right)(\ome{i} - \ome{i-1}).
\end{align*}
Applying Lemma \ref{III.L7} yields
\begin{align*}
r_0 \cdot \opt &\geq \ome{r}\cdot r_0 - (\ome{r}-\ome{0}) \sup_{i\in\phases} \left(\frac{\ca{i}}{r_0}\sum_{p\in\Paths_{s,t}} \fh_p\qpp{i}{p}\right) \\
&= \comp{f} \cdot r_0 - \comp{f} \sup_{i\in\phases} \left(\frac{\ca{i}}{r_0}\sum_{p\in\Paths_{s,t}} \fh_p\qpp{i}{p}\right),
\end{align*}
where we used $\ome{r} = \comp{f}$. Dividing both sides by $r_0 \cdot \comp{f}$ yields 
$$\frac{\opt}{\comp{f}} \geq 1 - \frac{1}{{r_0}^2} \sup_{i\in\phases} \ca{i}\sum_{p\in\Paths_{s,t}} \fh_p\qpp{i}{p}.$$ 
Since we do not know whether this supremum is attained we have to inject an $\epsilon$ error term here. We know that for any $\epsilon > 0$ there is a phase $i\in\phases$ such that
\begin{align*}
\frac{\opt}{\comp{f}} + \epsilon &\geq 1 - \frac{1}{{r_0}^2} \ca{i}\sum_{p\in\Paths_{s,t}} \fh_p\qpp{i}{p} \\
&=1 - \frac{\ca{i}}{{r_0}^2} \sum_{e\in E} \num{e}\qpp{i}{e},
\end{align*}
where we used $\qpp{i}{p} = \sum_{e\in p} \qpp{i}{e}$ and the fact that $\fh$ saturates every edge, especially that $\num{e} = \fh_e$ for all edges $e \in E$.
\end{proof}

\IIILeight*
Note, that Cominetti et al.~\cite{cominetti17} show (after~\cite{bhaskar15} was published) that there are networks where Nash flows over time route flow into the sink at rate higher than the inflow rate at the source during certain phases. But Lemma \ref{III.L8} and the corresponding result in~\cite{bhaskar15} only hold for phases of a Nash flow over time in which it holds that $\frac{r_0}{\ca{i}} \geq 1$. Nevertheless it is implicitly ensured that Nash flows over time on saturated graphs satisfy this in every phase since we lowered the edge-capacities to the value of $\fh$ with $|\fh|=r_0$.

\begin{proof}[Proof of Lemma \ref{III.L8}]
Since we concentrate on a single phase $i$ again we omit the subscript for the case, i.e., $\xp{e}, \lp{v}, c_v,$ and $\qp{e}$ all correspond to that phase $k$. Let $E^*$ be the set of edges with strict positive queues at the beginning of the phase and $E^\prime$ the set of edges active in phase $i$. By definition we have for every $e=uv\in E$
\begin{equation}\label{III.L8.1}
\text{$\qp{e}=$}
\left\{
\begin{aligned}
\lp{v} - \lp{u} ~ &\text{ if }e\in E^\prime, \\
0 \phantom{ - \lp{u} ~ } &\text{ otherwise.}\\
\end{aligned}
\right.
\end{equation}
With that and a case distinction over the $e\in E$ we prove $c_v \num{e} \qp{e} = \xp{e} \cdot (1- \frac{\lp{u}}{\lp{v}})$ for all $e\in E$. If $e\notin E^\prime$, then $e$ is not active in phase $i$ and therefore $\qp{e} = \xp{e} = 0$ and the statement is true in this case. If $e\in E^\prime\setminus E^*$, then $\xp{e}$ might be positive, but $e$ has no queue throughout the phase and thus $\qp{e} = 0$. But from Equation \eqref{III.L8.1} it follows that in this case it must hold that $\lp{u} = \lp{v}$ and therefore the statement also holds in this case. Finally if $e\in E^*$, i.e.\ $\qp{e} \neq 0$ then by definition of the thin flow $\xp{e} = c_v \num{e}\lp{v}$ and thus 
\begin{equation*}
\xp{e} \cdot \left(1-\frac{\lp{u}}{\lp{v}}\right) = c_v \num{e} \lp{v} \cdot \left(1- \frac{\lp{u}}{\lp{v}}\right) = c_v \num{e} \cdot \left(\lp{v}-\lp{u}\right) \overset{\eqref{III.L8.1}}{=} c_v \num{e} \qp{e}.
\end{equation*}

Let $\Paths_{s,t}$ be the set of all simple $s$-$t$ paths in $G$ and let $\{\xp{p}\}_{p\in\Paths_{s,t}}$ be a path decomposition of $\xp{~}$. Then by summing over all edges $e \in E$ we get the following bound from the above equation and the fact that $\xp{e} = \sum_{p\in\Paths_{s,t}:e\in p} \xp{p}$.
\begin{align*}
\sum_{e=uv\in E} c_v \num{e} \qp{e} &= \sum_{e=uv\in E} \xp{e} \cdot \left(1-\frac{\lp{u}}{\lp{v}}\right) \\
&= \sum_{p\in\Paths_{s,t}} \xp{p} \sum_{uv\in p} \left(1-\frac{\lp{u}}{\lp{v}}\right) \\
&\leq \sum_{p\in\Paths_{s,t}} \xp{p} \ln\left(\frac{r_0}{\ca{i}}\right) \\
&= r_0 \ln\left(\frac{r_0}{\ca{i}}\right),
\end{align*}
where in the inequality we used~\cite[Claim~12]{correa19} and in the last equality we used $|\xp{~}|=r_0$. Now let $\minc{f}{i} \coloneqq \min\{c_v(\theta) : v\in V, \theta \in (\theta_{i-1}, \theta_i)\}$, i.e.\ $\minc{f}{i}$ is the minimum of all spillback factors of every vertex $v$ in phase $i$ of the Nash flow $f$. The above inequality then yields 
\begin{equation*}
\minc{f}{i} \sum_{e\in E} \num{e} \qp{e} \leq \sum_{e=uv\in E} c_v \num{e} \qp{e} \leq r_0 \cdot \ln\left(\frac{r_0}{\ca{i}}\right).
\end{equation*}
And since we know that all spillback factors are strictly positive it follows that 
\begin{equation*}
\sum_{e\in E} \num{e} \qp{e} \leq \frac{r_0}{\minc{f}{i}} \cdot \ln\left(\frac{r_0}{\ca{i}}\right). \qedhere
\end{equation*}
\end{proof}
\fi

\end{document}